\documentclass[a4paper]{article}
\usepackage{spconf,amsmath,graphicx,cite}

\usepackage{amssymb}
\usepackage{amsthm}

\def\kron{\otimes}
\def\tr{\mathrm{tr}}

\def\diag{\mathrm{diag}}

\def\Htran{\mbox{\tiny H}}
\def\Ttran{\mbox{\tiny T}}

\newcommand{\fracSumtwo}[2]{\overset{#2}{\underset{#1}{\sum}}}
\newcommand{\vect}[1]{\mathbf{#1}}

\theoremstyle{remark}

\newtheorem{theorem}{Theorem}
\newtheorem{corollary}{Corollary}
\newtheorem{lemma}{Lemma}

\title{Distributed Massive MIMO in Cellular Networks: \\ Impact of Imperfect Hardware and Number of Oscillators}

\name{Emil Bj\"{o}rnson$^*$, Michail Matthaiou$^\dagger$$^\ddagger$, Antonios Pitarokoilis$^*$, and Erik G.~Larsson$^*$ %
	\thanks{$^*$ This research has received funding from ELLIIT, CENIIT, and the EU 7th Framework Programme under GA no ICT-619086 (MAMMOET).}%
}
\address{%
	$^*$ Dept. of Electrical Engineering (ISY), Link\"{o}ping University, Link\"{o}ping, Sweden\\
	$^\dagger$ ECIT Institute, Queen's University Belfast, Belfast, Northern Ireland, U.K.\\
	$^\ddagger$ Department of Signals and Systems, Chalmers University of Technology, Sweden \\
	\{emil.bjornson, antonios.pitarokoilis, erik.g.larsson\}@liu.se, m.matthaiou@qub.ac.uk\\
}

\begin{document}

\maketitle
\begin{abstract}
Distributed massive multiple-input multiple-output (MIMO) combines the array gain of coherent MIMO processing with the proximity gains of distributed antenna setups. In this paper, we analyze how transceiver hardware impairments affect the downlink with maximum ratio transmission. We derive closed-form spectral efficiencies expressions and study their asymptotic behavior as the number of the antennas increases. We prove a scaling law on the hardware quality, which reveals that massive MIMO is resilient to additive distortions, while multiplicative phase noise is a limiting factor. It is also better to have separate oscillators at each antenna than one per BS.
\end{abstract}

\enlargethispage{2mm}

\vspace{-3mm}

\section{Introduction}
\label{sec:intro}

\vspace{-2mm}

Cellular radio access networks (RANs) have conventionally consisted of one single-antenna base station (BS) per cell that served one user equipment (UE) per time-frequency resource. Since the increasing data traffic calls for higher spectral efficiencies [bit/symbol/cell], the RAN structure is now evolving to enable coherent downlink (DL) transmission to multiple UEs per resource symbol. The LTE-A standard has basic support for multi-user MIMO using a handful of co-located antennas. The massive MIMO concept from \cite{Marzetta2010a} takes multi-user MIMO to the 5G era by using hundreds of BS antennas to serve tens of UEs in parallel on each resource block. 

As the cellular concept is evolving, we can also question whether future BSs should be in the cell centers as in the past or distributed over the cells. The cloud RAN concept from \cite{CRAN2011} provides an efficient way to operate distributed antenna arrays and perform the coherent processing required by massive MIMO. While the majority of works on massive MIMO considers co-located arrays, the recent works \cite{Truong2013a,Yin2014a,Bjornson2015b} show that distributed massive MIMO can provide even higher spectral efficiencies than co-located deployments, due to proximity gains.

A potential showstopper for distributed massive MIMO would be if the technology is too sensitive to transceiver hardware impairments; for example, phase noise in local oscillators (LOs), amplifier non-linearities, non-ideal analog filters, and finite-precision analog/digital converters. The impact of hardware impairments on massive MIMO has received considerable attention in recent years \cite{Bjornson2014a,Bjornson2015b,Pitarokoilis2015a,Pitarokoilis2015b,Khanzadi2015a,Krishnan2015b,Gustavsson2014a}, but only \cite{Bjornson2015b} considered distributed arrays. The paper \cite{Bjornson2014a} showed that it is of fundamental importance to include hardware impairments in the performance analysis, since this can be a main limiting factor in systems with many antennas.
Nevertheless,  \cite{Bjornson2014a,Bjornson2015b} showed that massive MIMO is resilient to additive distortions originating from the BS. Multiplicative distortions such as phase noise can, however, hinder the system performance. These works use simplified stochastic impairment models, but the validity of the results has been confirmed  in \cite{Gustavsson2014a} by simulations based on sophisticated and realistic models.

For distributed arrays, an important question is whether the antennas should share a common LO (CLO) or if each antenna should be equipped with a separate LO (SLO). A number of recent works have looked into how this design choice impacts the severeness of the phase noise \cite{Bjornson2015b,Pitarokoilis2015a,Khanzadi2015a,Krishnan2015b,Pitarokoilis2015b}. The papers \cite{Bjornson2015b,Pitarokoilis2015a,Krishnan2015b,Pitarokoilis2015b} seem to establish the consensus that a setup with SLOs is preferable in the uplink (UL), since the independent phase rotations average out over the BS antennas. However, the answer is still open when it comes to the DL; \cite{Khanzadi2015a} showed that a CLO is preferable for non-fading channels, while \cite{Krishnan2015b}  considered fading single-cell systems and claimed that CLO prevails for few BS antennas (per user) or high SNR, and SLOs are desirable in the opposite cases.

In this paper, we extend our previous UL work in \cite{Bjornson2015b} to the DL. We consider a multi-cell massive MIMO system with distributed arrays and three kinds of hardware impairments: phase noise, distortion noise, and noise amplification. We derive new spectral efficiency expressions for maximum ratio transmission (MRT), which establish a  performance baseline in hardware-impaired multi-cell scenarios. These expressions are used to prove how the hardware quality may scale with the number of antennas. The analysis shows that SLOs is systematically a better choice than CLO also in the DL.

\vspace{-3mm}

\section{System Model}
\label{sec:system-model}

\vspace{-2mm}

We consider a cellular network with $L$ cells that operate in a synchronized time-division duplex (TDD) mode. Each cell serves $K$ single-antenna UEs using a BS equipped with $N$ antennas, which can be arbitrarily distributed over the coverage area. The TDD protocol divides the time-frequency resources into coherence blocks, as illustrated in Fig.~\ref{figure_protocol}. Each block consists of $T$ symbols with time indices $t = -\tau_{\mathrm{UL}}+1, \ldots, B+\tau_{\mathrm{DL}}$, whereof $\tau_{\mathrm{UL}}$ are UL data symbols, $B$ are UL pilots, and $\tau_{\mathrm{DL}}$ are DL data symbols. Note that $T = \tau_{\mathrm{UL}} \!+\! \tau_{\mathrm{DL}} \!+\! B$.

Let $(\cdot)^{\Ttran}$ and $(\cdot)^{\Htran}$  denote the transpose and conjugate transpose, respectively. The channel response between UE $k$ in cell $l$ and BS $j$ is a constant vector $\vect{h}_{jlk} \triangleq [h_{jlk}^{(1)} \, \ldots \, h_{jlk}^{(N)} ]^{\Ttran} \in \mathbb{C}^{N}$ within each block, where $h_{jlk}^{(n)}$ is the channel response for the $n$th BS antenna.
The channels are assumed to be Rayleigh fading as 
\begin{equation}
\vect{h}_{jlk} \sim \mathcal{CN}(\vect{0}, \vect{\Lambda}_{jlk} ),
\end{equation}
where the covariance matrix is $\vect{\Lambda}_{jlk} \triangleq \diag( \lambda_{jlk}^{(1)},\ldots,\lambda_{jlk}^{(N)} )$. The average channel attenuation $\lambda_{jlk}^{(n)} \geq 0$ is different for each combination of cell indices, UE index, and BS antenna index $n$. It depends, for example, on how the BS antennas are distributed in the cell and on the UE positions.

\begin{figure}
\begin{center}
\includegraphics[width=.95\columnwidth]{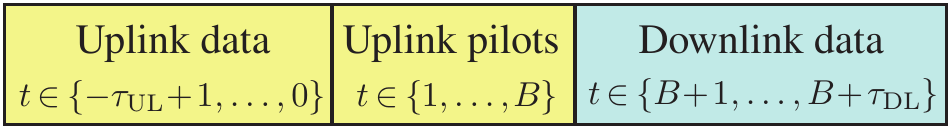}
\end{center}\vskip-6mm
\caption{Illustration of the TDD protocol where each coherence block consists of $T = \tau_{\mathrm{UL}} + \tau_{\mathrm{DL}} +B$ symbols.} \label{figure_protocol} \vskip-1mm
\end{figure}

\subsection{Uplink Model with Hardware Impairments}

A main goal of this paper is to investigate how transceiver hardware impairments impact the DL spectral efficiency.  We mainly consider impairments at the BSs, since the massive MIMO deployment constraints (e.g., cost, power, synchronization, and size restrictions) are likely to lead to BS hardware of lower quality than in contemporary networks.

To this end, we adopt the UL system model from \cite{Bjornson2015b} and generalize it to also cover the DL. Since the BSs in massive MIMO use channel estimates from the UL to perform transmit precoding in the DL, we need to model both directions of the links. As in \cite{Bjornson2015b}, the received UL signal $\vect{y}_j(t) \in \mathbb{C}^N$ in cell $j$ at symbol time $t \in \{ -\tau_{\mathrm{UL}}+1, \ldots, B\}$  is modeled as \vskip-4mm
\begin{equation} \label{eq:uplink-system-model}
\vect{y}_j(t) = \vect{D}_{\boldsymbol{\phi}_j(t)} \sum_{l=1}^{L} \vect{H}_{jl} \vect{x}_{l}(t) + \boldsymbol{\upsilon}_{j}(t) + \boldsymbol{\eta}_{j}(t)
\end{equation} \vskip-2mm
\noindent where $\vect{x}_{l}(t) = [x_{l1}(t) \, \ldots \, x_{lK}(t)]^{\Ttran} \in \mathbb{C}^{K}$ contains pilot/data symbols from UEs in cell $l$ and the channel matrix from these UEs to BS $j$ is
$\vect{H}_{jl} \triangleq [ \vect{h}_{jl1} \, \ldots \, \vect{h}_{jlK}] \in \mathbb{C}^{N \times K}$. The symbols from UE $k$ in cell $j$ have power $p_{jk}^{\mathrm{UL}} = \mathbb{E}\{ |x_{jk}(t)|^2\}$,
where $ \mathbb{E} \{ \cdot \}$ denotes the expected value of a random variable.

The matrix $\vect{D}_{\boldsymbol{\phi}_j(t)} \! \triangleq \! \diag\left(e^{\imath \phi_{j1}(t)},\ldots, e^{\imath \phi_{jN}(t)}\right)$ models the multiplicative effect of phase noise (with $\imath = \sqrt{-1}$). The variable $\phi_{jn}(t)$ is the phase rotation at the $n$th BS antenna in cell $j$ at time $t$, and it is modeled as a Wiener process \cite{Petrovic2007a}: 
$\phi_{jn}(t)  \sim  \mathcal{N}( \phi_{jn}(t - 1), \delta)$
where $\delta \geq 0$ is the variance of the phase-noise increments. We consider two implementations: 
\begin{enumerate}
\item Common LO (CLO): $\phi_{j1}(t) \!=\! \ldots \!=\! \phi_{jN}(t)$ within a cell.
\item Separate LOs (SLOs): All $ \phi_{jn}(t)$ are independent.
\end{enumerate}
The above represent having one LO that feeds all antennas at BS $j$ or one separate LO connected to each of the $N$ antennas.

Moreover, $\boldsymbol{\upsilon}_{j}(t) \sim \mathcal{CN}(\vect{0},\vect{\Upsilon}_j(t) )$ is additive distortion noise (e.g., from finite-precision quantization, non-linearities, and interference leakage in the frequency domain). It is 
proportional to the received signal power at the antenna and uncorrelated between antennas \cite{Bjornson2015b,Zhang2012a}:
\begin{equation} \label{eq:distortion-uplink}
\vect{\Upsilon}_j(t) \triangleq \kappa_{\mathrm{UL}}^2 \sum_{l=1}^{L} \sum_{k=1}^{K} p_{lk}^{\mathrm{UL}} \diag\bigg( |h_{jlk}^{(1)}|^2,\ldots, |h_{jlk}^{(N)}|^2\bigg)
\end{equation}
where $\kappa_{\mathrm{UL}} \geq 0$ is the proportionality coefficient.

Finally, $\boldsymbol{\eta}_{j}(t) \sim \mathcal{CN}(\vect{0},\sigma_{\mathrm{BS}}^2 \vect{I}_N)$ is the receiver noise with variance $\sigma_{\mathrm{BS}}^2$ (including noise amplification in circuits).

\subsection{Downlink Model with Hardware Impairments}

Similar to the UL, we model the received DL signal $z_{jk}(t) \in \mathbb{C}$ at UE $k$ in cell $j$ at  time $t \in \{B+1,\ldots,B+\tau_{\mathrm{DL}} \}$ as
\begin{equation} \label{eq:downlink-system-model}
z_{jk}(t) = \!\sum_{l=1}^{L} \vect{h}_{ljk}^{\Htran} \! \left( \! \vect{D}_{\boldsymbol{\phi}_l(t)} \!\sum_{m=1}^{K} \vect{w}_{lm}(t) s_{lm}(t)
+ \boldsymbol{\psi}_{l}(t) \! \right) + \eta_{jk}(t)
\end{equation}
where $s_{lm}(t)$ is the DL data symbol (with power $p_{jk}^{\mathrm{DL}} = \mathbb{E}\{ |s_{lm}(t)|^2\}$) and $\vect{w}_{lm}(t) \triangleq [w_{lm}^{(1)}(t) \, \ldots \, w_{lm}^{(N)}(t) ]^{\Ttran} \in \mathbb{C}^{N}$ is the corresponding linear precoding vector. The receiver noise is $\eta_{jk}(t) \sim \mathcal{CN}(0,\sigma_{\mathrm{UE}}^2)$, where $\sigma_{\mathrm{UE}}^2$ is the variance (including noise amplification). The phase-noise matrix $\vect{D}_{\boldsymbol{\phi}_j(t)}$ was defined earlier, while $\boldsymbol{\psi}_{j}(t) \! \sim \! \mathcal{CN}(\vect{0},\vect{\Psi}_j )$ is the additive distortion in the DL (e.g., due to non-linearities and leakage in the frequency domain). Similar to \eqref{eq:distortion-uplink}, the distortion at a certain antenna is proportional to the transmit power at this antenna and uncorrelated with the distortions at other antennas: 
\begin{equation*}
 \vect{\Psi}_j \triangleq \kappa_{\mathrm{DL}}^2  \sum_{k=1}^{K} p_{jk}^{\mathrm{DL}} \diag\left( |w_{jk}^{(1)}(t)|^2,\ldots, |w_{jk}^{(N)}(t)|^2\right)\notag
\end{equation*}
where $\kappa_{\mathrm{DL}} \geq 0$ is the proportionality coefficient.

This system model is used in the next section to compute achievable DL spectral efficiencies. These depend on the level of hardware impairments, as characterized by the variance of the phase-noise increments $\delta$, the distortion noise proportionality coefficients $\kappa_{\mathrm{UL}},\kappa_{\mathrm{DL}}$, and the receiver noise variances $\sigma_{\mathrm{BS}}^2,\sigma_{\mathrm{UE}}^2$.
The results are applicable for any $p_{jk}^{\mathrm{DL}}$ and $p_{jk}^{\mathrm{UL}}$, for each $j$ and $k$, thus under arbitrary power control.

\section{Downlink Performance Analysis}

In this section, we derive the DL spectral efficiency per UE and study its asymptotic behavior (when $N$ is large) to understand the impact of hardware impairments.

\begin{figure*}[t!]
\begin{align} \label{eq:achievable-SINR}
\mathrm{SINR}_{jk}(t) =  \frac{ p_{jk}^{\mathrm{DL}} \frac{| \mathbb{E}\{  \vect{h}_{jjk}^{\Htran}(t)  \boldsymbol{\omega}_{jk}(t)    \} |^2}{ \mathbb{E}\{ \| \boldsymbol{\omega}_{jk}(t) \|^2\} } }{ \fracSumtwo{l=1}{L} \fracSumtwo{m=1}{K} p_{lm}^{\mathrm{DL}}  \left( \frac{ \mathbb{E}\{ | \vect{h}_{ljk}^{\Htran}(t)  \boldsymbol{\omega}_{lm}(t)    |^2 \}  + \kappa_{\mathrm{DL}}^2   \sum_{n=1}^{N} \mathbb{E}\{ | \vect{h}_{ljk}^{(n)}|^2 | \boldsymbol{\omega}_{lm}^{(n)}(t)    |^2 \} }{ \mathbb{E}\{ \| \boldsymbol{\omega}_{lm}(t) \|^2\} }  
\right)
- p_{jk}^{\mathrm{DL}} \frac{| \mathbb{E}\{  \vect{h}_{jjk}^{\Htran}(t)  \boldsymbol{\omega}_{jk}(t)    \} |^2}{ \mathbb{E}\{ \| \boldsymbol{\omega}_{jk}(t) \|^2\} }  +  \sigma_{\mathrm{UE}}^2 } \tag{13}
\end{align} \vskip-2mm
\hrulefill
\vskip-3mm
\end{figure*}

\subsection{Uplink Channel Estimation}

In order to perform coherent transmit precoding in the DL, each BS acquires the channels to its UEs by using the UL pilots. The pilot sequence of UE $k$ in cell $j$ is defined as $\tilde{\vect{x}}_{jk} \triangleq [x_{jk}(1) \, \ldots \, x_{jk}(B)]^{\Ttran} \in \mathbb{C}^{B \times 1}$. The analysis in this paper holds for arbitrary pilot sequences (with $|x_{jk}(b)|^2 = p_{jk}^{\mathrm{UL}}$ for $b=1,\ldots,B$), while we consider columns from a Fourier matrix in Sec.~\ref{eq:numerical-results} (to achieve mutual orthogonality and constant energy per symbol).
Since the effective channels
\begin{equation}
\vect{h}_{jlk}(t) \triangleq \vect{D}_{\boldsymbol{\phi}_j(t)} \vect{h}_{jlk}
\end{equation}
depend on the phase-noise and are different at every symbol time $t$, we need a channel estimator that provides new estimates at each $t$. Such a linear minimum mean-squared error (MMSE) estimator was derived in \cite{Bjornson2015b} and is described below, since the expressions are used in the forthcoming analysis.

\begin{lemma}[Th.~1 in \cite{Bjornson2015b}] \label{lemma:LMMSE-estimation}
Let $\boldsymbol{\psi}_j \triangleq \left[\vect{y}_j^{\Ttran}(1) \, \ldots \, \vect{y}_j^{\Ttran}(B)\right]^{\Ttran} \in \mathbb{C}^{BN}$ denote the combined received signal in cell $j$ from the pilot transmission. The LMMSE estimate of $\vect{h}_{jlk}(t)$ at any symbol time $t \in \{ -\tau_{\mathrm{UL}}+1, \ldots, B+\tau_{\mathrm{DL}} \}$ for any $l$ and $k$ is
\begin{equation} \label{eq:LMMSE-estimator}
  \hat{\vect{h}}_{jlk}(t) = \left( \tilde{\vect{x}}_{lk}^{\Htran}  \vect{D}_{\boldsymbol{\delta}(t)} \kron \vect{\Lambda}_{jlk} \right) \boldsymbol{\Phi}^{-1}_j \boldsymbol{\psi}_j
\end{equation}
where $\kron$ denotes the Kronecker product, 
\begin{align}
\vect{D}_{\boldsymbol{\delta}(t)} & \triangleq \diag\bigg( e^{-\frac{\delta}{2} |t-1|}, \ldots, e^{-\frac{\delta}{2} |t-B|}\bigg), \\
\boldsymbol{\Phi}_j &\triangleq \sum_{\ell=1}^{L} \sum_{m=1}^{K} \vect{X}_{\ell m} \kron \vect{\Lambda}_{j \ell m}  + \sigma_{\mathrm{BS}}^2 \vect{I}_{BN},
\end{align}
and the element at position $(b_1,b_2)$ in $\vect{X}_{\ell m} \in \mathbb{C}^{B \times B}$ is 
\begin{equation}
[\vect{X}_{\ell m} ]_{b_1,b_2} = \begin{cases} p_{\ell m}^{\mathrm{UL}} (1 + \kappa_{\mathrm{UL}}^2) , & b_1 = b_2, \\ x_{\ell m}(\tau_{b_1}) x_{\ell m}^*(\tau_{b_2}) e^{-\frac{\delta}{2} |\tau_{b_1}-\tau_{b_2}|}, &  b_1 \neq b_2. \end{cases}
\end{equation}
The corresponding error covariance matrix is
\begin{align*} 
\vect{C}_{jlk}(t) =
\vect{\Lambda}_{jlk}\! -\! ( \tilde{\vect{x}}_{lk}^{\Htran}  \vect{D}_{\boldsymbol{\delta}(t)} \kron \vect{\Lambda}_{jlk} ) \boldsymbol{\Phi}^{-1}_j
(  \vect{D}_{\boldsymbol{\delta}(t)}^{\Ttran} \tilde{\vect{x}}_{lk}   \kron \vect{\Lambda}_{jlk} ).
\end{align*}
\end{lemma}
\begin{proof}
This lemma follows from adapting results in \cite{Bjornson2015b} to the notation and power constraints considered herein.
\end{proof}

\vspace{-3mm}

\subsection{Downlink Spectral Efficiency}

Next, we derive achievable DL spectral efficiencies, using normalized linear precoding vectors of the general form
\begin{equation}
\vect{w}_{jk}(t) = \frac{ \boldsymbol{\omega}_{jk}(t) } {  \sqrt{ \mathbb{E} \{ \| \boldsymbol{\omega}_{jk}(t) \|^2  \}}}.
\end{equation}
With this notation, MRT is given by $\boldsymbol{\omega}_{jk}(t) = \hat{\vect{h}}_{jjk}(t)$.

\newpage

\begin{lemma} \label{lemma:achievable-rates}
Suppose that UE $k$ in cell $j$ knows the channel and interference statistics, but not the channel realizations.
An achievable lower bound on the ergodic capacity of this UE is
\begin{equation} \label{eq:achievable-rate}
R_{jk} = \frac{1}{T}  \sum_{t = B+1}^{B+\tau_{\mathrm{DL}}} \log_2 \big( 1 + \mathrm{SINR}_{jk}(t) \big) \quad [\textrm{bit/symbol}]
\end{equation}
where $\mathrm{SINR}_{jk}(t)$ is given in \eqref{eq:achievable-SINR} at the top of this page.
\end{lemma}
\begin{proof}
As in \cite{Pitarokoilis2015a}, we compute one spectral efficiency for each $t \in \{B \!+\! 1, B \!+ \! \tau_{\mathrm{DL}} \}$ since the effective channels vary with $t$. The expression is obtained by using the signal received over the average  channel  $\mathbb{E}\{  \vect{h}_{jjk}^{\Htran}(t)  \boldsymbol{\omega}_{jk}(t)  \} $ for decoding, while treating the signal received over the uncorrelated deviation from this average value, the inter-user interference and distortion noise as worst-case Gaussian noise in the decoder.
\end{proof}

\setcounter{equation}{13}

The expression in Lemma \ref{lemma:achievable-rates} is a reasonable bound on the practical performance that can be achieved using simple signal processing at the UE (i.e., detect the useful signal and treat everything unknown as Gaussian noise). The SINR expression in \eqref{eq:achievable-SINR} contains a number of expectations that can be computed numerically for any choice of precoding vectors. Next, we provide closed-form expressions for MRT.

\begin{theorem}
If MRT is used, then the expectations in $ \mathrm{SINR}_{jk}(t)$ of Lemma \ref{lemma:achievable-rates} are computed as in \eqref{eq:MRT-squared-norm}--\eqref{eq:MRT-second-moment-and-dist} at the top of the next page (where $\vect{e}_n$ denotes the $n$th column of $\vect{I}_{N}$).
\end{theorem}
\begin{proof}
Follows from straightforward computation of the expectations, whereof some are the same as in \cite[Th.~2]{Bjornson2015b}.
\end{proof}

\begin{figure*}[t!]
\begin{align}
\mathbb{E}\{ \| \boldsymbol{\omega}_{jk}(t) \|^2\} & = \mathbb{E}\{  \vect{h}_{jjk}^{\Htran}(t)  \boldsymbol{\omega}_{jk}(t)    \}  = \tr \left( \left( \tilde{\vect{x}}_{jk}^{\Htran}  \vect{D}_{\boldsymbol{\delta}(t)} \kron \vect{\Lambda}_{jjk} \right) \boldsymbol{\Phi}^{-1}_j \left(  \vect{D}_{\boldsymbol{\delta}(t)}^{\Ttran} \tilde{\vect{x}}_{jk}   \kron \vect{\Lambda}_{jjk} \right) \right)  \label{eq:MRT-squared-norm} \\
\label{eq:MRT-second-moment-and-dist}
\mathbb{E}\{ | \vect{h}_{ljk}^{\Htran}(t)  \boldsymbol{\omega}_{lm}(t)    |^2 \}
 &  +  \kappa_{\mathrm{DL}}^2  \sum_{n=1}^{N} \mathbb{E}\{ | \vect{h}_{ljk}^{(n)}|^2 | \boldsymbol{\omega}_{lm}^{(n)}(t)    |^2 \}  \\ &= (1+ \kappa_{\mathrm{DL}}^2) \tr \left( \vect{\Lambda}_{ljk} \left( \tilde{\vect{x}}_{lm}^{\Htran}  \vect{D}_{\boldsymbol{\delta}(t)} \kron \vect{\Lambda}_{llm} \right) \boldsymbol{\Phi}^{-1}_l \left(  \vect{D}_{\boldsymbol{\delta}(t)}^{\Ttran} \tilde{\vect{x}}_{lm}   \kron \vect{\Lambda}_{llm} \right) \right) \notag \\
&\!\!\!\!\!\!\!\!\!\!\!\!\!\!\!\!\!\!\!\!\!\!\!\!\!\!\!\!\!\!\!\!\!\!\!\!\!\!\!\!\!\!\!\!\!\!\!\!\!\!\!\!\!\!\!\! + \begin{cases}
\fracSumtwo{n_1=1}{N} \fracSumtwo{n_2=1}{N} \lambda_{llm}^{(n_1)} \lambda_{ljk}^{(n_1)} \lambda_{llm}^{(n_2)}  \lambda_{ljk}^{(n_2)}
 \left( \tilde{\vect{x}}_{lm}^{\Htran}  \vect{D}_{\boldsymbol{\delta}(t)} \kron \vect{e}_{n_1}^{\Htran} \right) \boldsymbol{\Phi}^{-1}_l \left( (\vect{X}_{jk} - \kappa_{\mathrm{UL}}^2 p_{jk}^{\mathrm{UL}} \vect{I}_B) \kron \vect{e}_{n_1} \vect{e}_{n_2}^{\Htran} \right) \boldsymbol{\Phi}^{-1}_l \left(  \vect{D}_{\boldsymbol{\delta}(t)}^{\Ttran}  \tilde{\vect{x}}_{lm} \kron \vect{e}_{n_2} \right) & \text{if CLO}\\
\left( \tr \left( \left( \tilde{\vect{x}}_{lm}^{\Htran}  \vect{D}_{\boldsymbol{\delta}(t)} \kron \vect{\Lambda}_{llm} \right) \boldsymbol{\Phi}^{-1}_l \left(  \vect{D}_{\boldsymbol{\delta}(t)}^{\Ttran} \tilde{\vect{x}}_{jk}   \kron \vect{\Lambda}_{ljk} \right) \right)   \right)^2 & \text{if SLOs}
\end{cases} \notag \\
&\!\!\!\!\!\!\!\!\!\!\!\!\!\!\!\!\!\!\!\!\!\!\!\!\!\!\!\!\!\!\!\!\!\!\!\!\!\!\!\!\!\!\!\!\!\!\!\!\!\!\!\!\!\!\!\! + \begin{cases}
\fracSumtwo{n=1}{N}  \left( \lambda_{llm}^{(n)} \lambda_{ljk}^{(n)}  \right)^2
\left( \tilde{\vect{x}}_{lm}^{\Htran}  \vect{D}_{\boldsymbol{\delta}(t)} \kron \vect{e}_n^{\Htran} \right) \boldsymbol{\Phi}^{-1}_l
\left( ( \kappa_{\mathrm{UL}}^2 p_{jk}^{\mathrm{UL}} \vect{I}_B + \kappa_{\mathrm{DL}}^2 \vect{X}_{jk} )
 \kron \vect{e}_n \vect{e}_n^{\Htran} \right) \boldsymbol{\Phi}^{-1}_l \left(  \vect{D}_{\boldsymbol{\delta}(t)}^{\Ttran}  \tilde{\vect{x}}_{lm} \kron \vect{e}_n \right) & \!\!\text{if CLO}\\
\fracSumtwo{n=1}{N}  \left( \lambda_{llm}^{(n)} \lambda_{ljk}^{(n)}  \right)^2
\left( \tilde{\vect{x}}_{lm}^{\Htran}  \vect{D}_{\boldsymbol{\delta}(t)} \kron \vect{e}_n^{\Htran} \right) \boldsymbol{\Phi}^{-1}_l
\left( ( (1+ \kappa_{\mathrm{DL}}^2) \vect{X}_{jk} - \vect{D}_{\boldsymbol{\delta}(t)}^{\Ttran} \tilde{\vect{x}}_{jk}   \tilde{\vect{x}}_{jk}^{\Htran} \vect{D}_{\boldsymbol{\delta}(t)}) \kron \vect{e}_n \vect{e}_n^{\Htran} \right) \boldsymbol{\Phi}^{-1}_l \left(  \vect{D}_{\boldsymbol{\delta}(t)}^{\Ttran}  \tilde{\vect{x}}_{lm} \kron \vect{e}_n \right)
& \!\!\text{if SLOs}
\end{cases} \notag 
\end{align} \vskip-2mm
\hrulefill
\vskip-5mm
\end{figure*}

\vspace{-4mm}

\subsection{Asymptotic Behavior and Scaling Laws}

Next, we investigate the behavior at large $N$. For tractability, we consider $A < \infty$ spatially separated subarrays each with $\frac{N}{A}$ antennas. Recall that these antennas are either controled by a common LO that sends clock signals or separate LOs at each antenna. The channel covariance matrices then factorize as 
\begin{equation} \label{eq:covariance-matrices-asymptotics}
\vect{\Lambda}_{jlk} = \tilde{\vect{\Lambda}}_{jlk}^{(A)} \kron \vect{I}_{\frac{N}{A}}
\end{equation}
where $ \tilde{\vect{\Lambda}}_{jlk}^{(A)} = \diag( \tilde{\lambda}_{jlk}^{(1)},\ldots,\tilde{\lambda}_{jlk}^{(A)}) \in \mathbb{C}^{A \times A} $ and $ \tilde{\lambda}_{jlk}^{(a)}$ is the average channel attenuation between subarray $a$ in cell $j$ and UE $k$ in cell $l$. By letting the number of antennas in each subarray grow large, we obtain the following property. 

\begin{corollary} \label{corollary:asymptotic-SINR}
If MRT is used and the channel covariance matrices can be factorized as in \eqref{eq:covariance-matrices-asymptotics}, then
\begin{equation} \label{eq:asymptotic-SINR}
\mathrm{SINR}_{jk}(t) = \frac{ p_{jk}^{\mathrm{DL}} \mathcal{S}_{jk} }{\fracSumtwo{l=1}{L} \fracSumtwo{m=1}{K} p_{lm}^{\mathrm{DL}} \mathcal{I}_{lmjk}   - p_{jk}^{\mathrm{DL}} \mathcal{S}_{jk} \!+\! \mathcal{O}(\frac{1}{N}) }
\end{equation}
where the signal part is \vskip-5mm
\begin{equation*}
\mathcal{S}_{jk} = \tr \left( ( \tilde{\vect{x}}_{jk}^{\Htran}  \vect{D}_{\boldsymbol{\delta}(t)} \!\kron\! \tilde{\vect{\Lambda}}_{jjk}^{(A)} ) \widetilde{\boldsymbol{\Phi}}^{-1}_j (  \vect{D}_{\boldsymbol{\delta}(t)} \tilde{\vect{x}}_{jk}  \!\kron\! \tilde{\vect{\Lambda}}_{jjk}^{(A)} ) \right)
\end{equation*} \vskip-2mm
\noindent with $\widetilde{\boldsymbol{\Phi}}_j \triangleq \sum_{\ell=1}^{L} \sum_{m=1}^{K} \vect{X}_{\ell m} \kron \tilde{\vect{\Lambda}}_{j \ell m}^{(A)}  + \sigma_{\mathrm{BS}}^2 \vect{I}_{AB}$, where
the interference terms $\mathcal{I}_{lmjk}$ with a CLO are \vskip-5mm
\begin{align*}
&\mathcal{I}_{lmjk}^{\mathrm{CLO}} = \frac{  \!\fracSumtwo{a_1=1}{A} \fracSumtwo{a_2=1}{A} \tilde{\lambda}_{llm}^{(a_1)} \tilde{\lambda}_{ljk}^{(a_1)} \tilde{\lambda}_{llm}^{(a_2)} \tilde{\lambda}_{ljk}^{(a_2)} \left( \tilde{\vect{x}}_{lm}^{\Htran}  \vect{D}_{\boldsymbol{\delta}(t)} \!\kron\! \vect{e}_{a_1}^{\Htran} \right)  }{ \tr \left( \left( \tilde{\vect{x}}_{lm}^{\Htran}  \vect{D}_{\boldsymbol{\delta}(t)} \kron \vect{\Lambda}_{llm} \right) \widetilde{\boldsymbol{\Phi}}^{-1}_l (  \vect{D}_{\boldsymbol{\delta}(t)}^{\Ttran} \tilde{\vect{x}}_{jk}   \kron \vect{\Lambda}_{ljk} ) \right)} \\
& \! \times \!  \widetilde{\boldsymbol{\Phi}}^{-1}_l \! \left( (\vect{X}_{jk} \!-\! \kappa_{\mathrm{UL}}^2 p_{jk}^{\mathrm{UL}} \vect{I}_B)  \kron \vect{e}_{a_1} \vect{e}_{a_2}^{\Htran} \right) \widetilde{\boldsymbol{\Phi}}^{-1}_l \! (  \vect{D}_{\boldsymbol{\delta}(t)}^{\Ttran}  \tilde{\vect{x}}_{lm} \!\kron\! \vect{e}_{a_2} ) \notag
\end{align*} \vskip-1mm
\noindent and the interference terms with SLOs are \vskip-5mm
\begin{align*}
&\mathcal{I}_{lmjk}^{\mathrm{SLOs}} = \!
 \tr \left( \left( \tilde{\vect{x}}_{lm}^{\Htran}  \vect{D}_{\boldsymbol{\delta}(t)} \kron \vect{\Lambda}_{llm} \right) \widetilde{\boldsymbol{\Phi}}^{-1}_l  (  \vect{D}_{\boldsymbol{\delta}(t)}^{\Ttran} \tilde{\vect{x}}_{jk}   \kron \vect{\Lambda}_{ljk} ) \right).
\end{align*}
The notation $\mathcal{O}(\frac{1}{N})$ is used for terms that go to zero as $\frac{1}{N}$ or faster when $N \rightarrow \infty$, while $\vect{e}_{a}$ is the $a$th column of $\vect{I}_A$.
\end{corollary}
\begin{proof}
Follows from dividing all terms in $\mathrm{SINR}_{jk}(t)$ by $\frac{N}{A}$ and then analyze the  expressions for MRT in \eqref{eq:MRT-squared-norm}--\eqref{eq:MRT-second-moment-and-dist}.
\end{proof}

This corollary does not contain $\kappa_{\mathrm{UL}}$, $\kappa_{\mathrm{DL}}$, $\sigma_{\mathrm{BS}}^2$, or $\sigma_{\mathrm{UE}}^2$, thus it shows that the impact of distortion noise and receiver noise vanishes as $N \rightarrow \infty$. The asymptotic SINRs are only limited by the channel distributions, pilot-contaminated interference, and phase noise. This  means that distributed massive MIMO systems can handle larger additive distortions than conventional systems, as manifested by the next corollary.

\begin{corollary} \label{cor:scaling-law}
Suppose that 
$\kappa_{\mathrm{UL}}^2 \!=\! \kappa_{\mathrm{UL},0}^2 N^{z_1}$, 
$\kappa_{\mathrm{DL}}^2 \!=\! \kappa_{\mathrm{DL},0}^2 N^{z_1}$, 
$\sigma_{\mathrm{BS}}^2 \!=\! \sigma_{\mathrm{BS},0}^2 N^{z_2}$, 
$\sigma_{\mathrm{UE}}^2 \!=\! \sigma_{\mathrm{UE},0}^2 N^{z_2}$, 
and $\delta \!=\! \delta_{0} (1+ \ln(N^{z_3}) )$, for some scaling exponents $z_1,z_2,z_3 \geq 0$ and  constants $\kappa_{\mathrm{UL},0},\kappa_{\mathrm{DL},0},\sigma_{\mathrm{BS},0}^2,\sigma_{\mathrm{UE},0}^2,\delta_0 \geq 0$.
The SINRs, $\mathrm{SINR}_{jk}(t)$, with MRT converge to non-zero limits as $N \rightarrow \infty$ if
\begin{equation} \label{eq:scaling-law}
\begin{cases} \max(z_1,z_2) \leq \frac{1}{2} \,\,\, \textrm{and} \,\,\, z_3=0 & \textrm{for a CLO} \\
\max(z_1,z_2) + z_3  \frac{\delta_{0} |\tau_{\mathrm{DL}}-B | }{2}   \leq \frac{1}{2} & \textrm{for SLOs}.
\end{cases}
\end{equation}
\end{corollary}
\begin{proof}
Based on \eqref{eq:MRT-squared-norm}--\eqref{eq:MRT-second-moment-and-dist}; see \cite{Bjornson2015b} for a similar proof.
\end{proof}

This corollary shows that the DL can handle additive distortions with variances that scale as $\sqrt{N}$ (i.e., $z_1=z_2 = \frac{1}{2}$), while achieving decent performance. The scaling law also shows that the phase noise variance with SLOs can increase logarithmically with $N$, while this is not allowed with a CLO. This proves that massive MIMO with SLOs are preferable in the DL, at least when the number of antennas is large. The scaling law holds also for precoders that are better than MRT.

\vspace{-2mm}

\section{Numerical Results}
\label{eq:numerical-results}

The analytic results are corroborated for the distributed massive MIMO setup in Fig.~\ref{figure_distributedscenario}. This is a wrap-around topology with 16 cells of $400 \times 400$ meters, each consisting of $A=4$ subarrays with $\frac{N}{A}$ antennas located 100 meters from the cell center. The $K=15$ UEs per cell are uniformly distributed, with a minimum distance of 25 meters from the subarrays. The transmit powers are $p_{jk}^{\mathrm{DL}} = p_{jk}^{\mathrm{UL}} = - 50$ dBm/Hz for all $j$ and $k$ (e.g., $100$ mW over 10 MHz). The channel attenuations are modeled as in \cite{Bjornson2015b}: $\lambda_{jlk}^{(n)} = 10^{s_{jlk}^{(n)}-1.53} /(d_{jlk}^{(n)})^{3.76}$, where $d_{jlk}^{(n)}$ is the distance in meters between BS antenna $n$ in cell $j$ and UE $k$ in cell $l$ and $s_{jlk}^{(n)} \sim \mathcal{N}(0,3.16)$ is shadow-fading.

The hardware impairments are characterized by the distortion proportionality coefficients $\kappa_{\mathrm{UL}} = \kappa_{\mathrm{DL}} = 0.03$, the variance of phase noise increments $\delta = 1 \cdot 10^{-5}$, and the receiver noise powers $\sigma_{\mathrm{BS}}^2 = \sigma_{\mathrm{UE}}^2 = -169$ dBm/Hz (with 5 dB noise amplification). These are also the initial constants when we scale the hardware quality based on Corollary \ref{cor:scaling-law}.

Fig.~\ref{figure3} shows the average spectral efficiency per UE. The coherence block contains $T = 300$ symbols, whereof $B=15$ symbols are used for pilot sequences and $\tau_{\mathrm{DL}} = 285$ for DL payload data. Hardware impairments incur a performance loss as compared to ideal hardware. The gap is small with SLOs, but larger with a CLO. This validates the analytic observation that SLOs is the better choice in massive MIMO.

The figure also illustrates the scaling law established by Corollary \ref{cor:scaling-law}. The middle curves show the behavior when satisfying the scaling law ($z_1 = z_2= 0.48$ with a CLO and adding also $z_3 = 0.48$ with SLOs). 
By gradually degrading the hardware with $N$, there is a performance loss at every $N$, but the curves are still increasing  with $N$. The performance loss is small for SLOs, but very large for a CLO. The curves at the bottom are for a case when the scaling law is not satisfied, which gives a performance that goes to zero as $N \rightarrow \infty$.

\begin{figure}
\begin{center} \vskip-2mm
\includegraphics[width=.95\columnwidth]{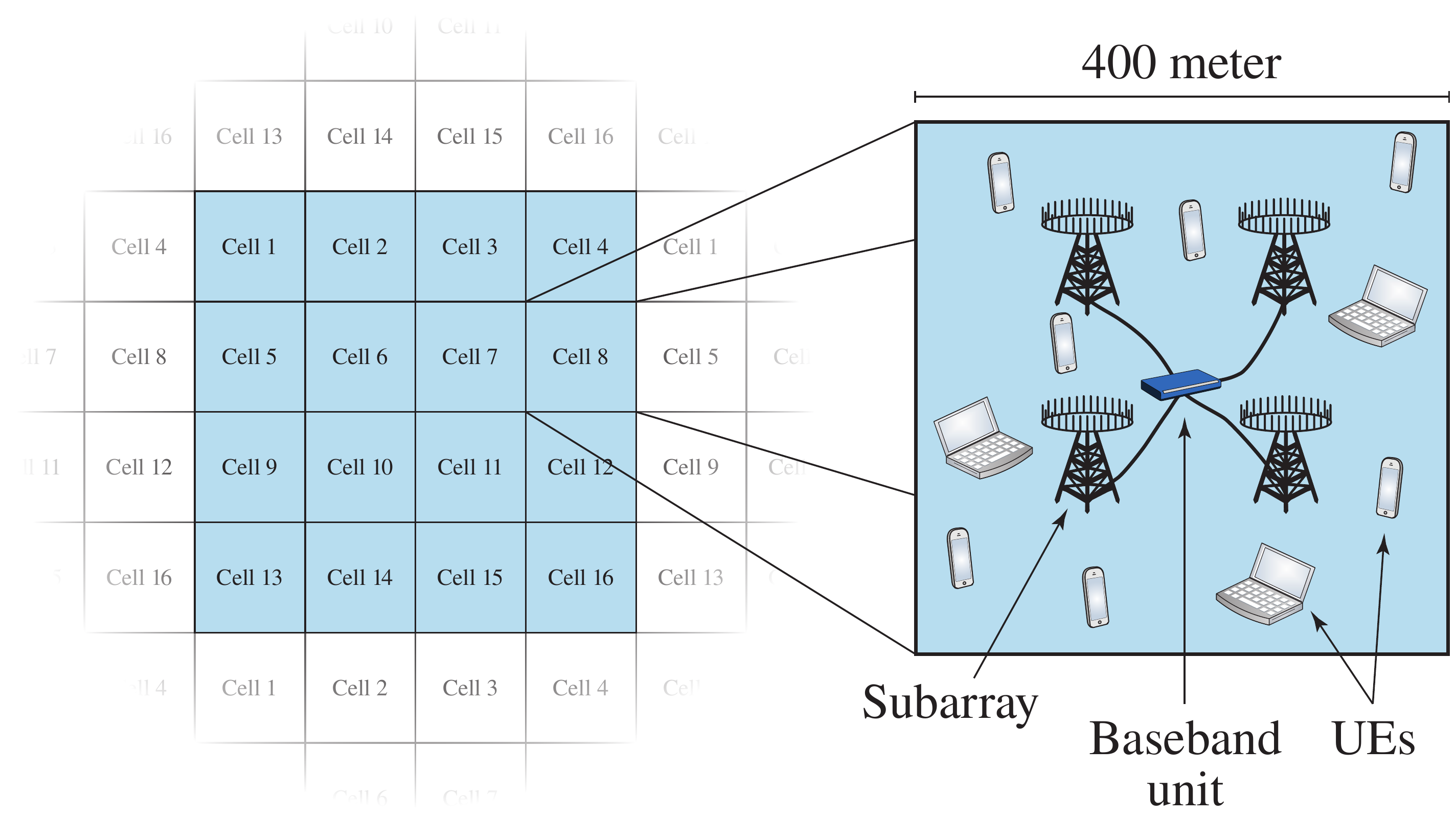}
\end{center}\vskip-8mm
\caption{Illustration of the multi-cell distributed massive MIMO scenario considered in the numerical evaluation.} \label{figure_distributedscenario} \vskip-3mm
\end{figure}

\vspace{-1mm}

\section{Concluding Remarks}

We have analyzed the DL performance of distributed massive MIMO systems, with focus on the impact of hardware impairments. We have proved that additive distortions have smaller impact on massive MIMO than conventional networks, since the variance may increase as $\sqrt{N}$ with little performance loss. Multiplicative phase noise can be more severe, but the performance is better if each BS antenna has a separate oscillator.

The DL analytic results in this paper are in line with previous UL results  in \cite{Bjornson2015b,Pitarokoilis2015a,Krishnan2015b,Pitarokoilis2015b}. This is natural since the UL-DL duality for systems with linear processing implies that the same performance is achievable in both directions (if the power allocation is optimized). However, our results stand in contrast to the recent works \cite{Khanzadi2015a,Krishnan2015b} where the DL behave differently than the UL when it comes to phase noise. This is due to different system models: \cite{Khanzadi2015a} considers high SNRs in non-fading single-user cases, while \cite{Krishnan2015b} considers a single cell with relatively good CSI. In comparison, we consider a generalized multi-cell setup with more inter-user interference and thus lower SINRs.

\vspace{-3mm}

\begin{figure}
\begin{center} 
\includegraphics[width=\columnwidth]{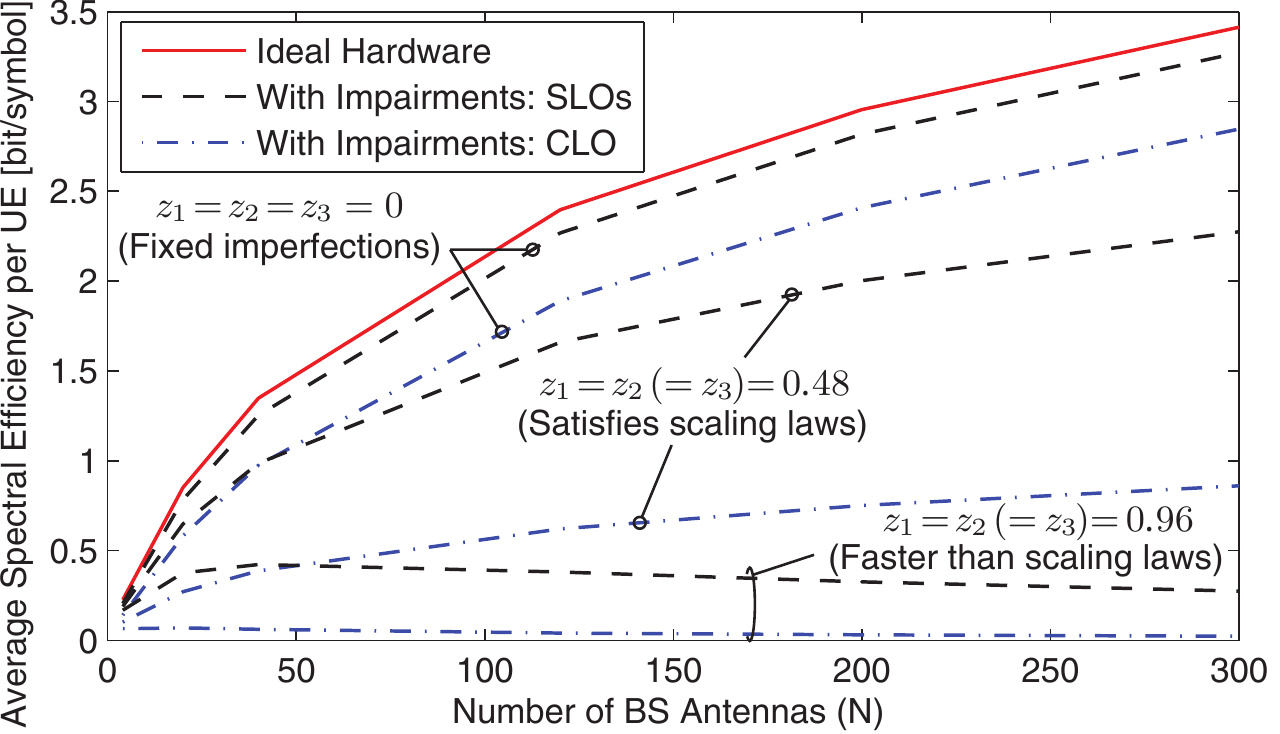}
\end{center}\vskip-6mm
\caption{Average DL spectral efficiency for distributed massive MIMO with fixed or increasing hardware impairments.} \label{figure3} \vskip-4mm
\end{figure}

\bibliographystyle{IEEEbib}
\bibliography{IEEEabrv,refs}

\begin{thebibliography}{10}

\bibitem{Marzetta2010a}
T.~L. Marzetta,
\newblock ``Noncooperative cellular wireless with unlimited numbers of base
  station antennas,''
\newblock {\em {IEEE} Trans. Wireless Commun.}, vol. 9, no. 11, pp. 3590--3600,
  2010.

\bibitem{CRAN2011}
China Mobile~Research Institute,
\newblock ``{C-RAN}: The road towards green {RAN},''
\newblock {\em White Paper}, Oct. 2011.

\bibitem{Truong2013a}
K.~T. Truong and R.W.~Heath Jr.,
\newblock ``The viability of distributed antennas for massive {MIMO} systems,''
\newblock in {\em Proc.~Asilomar CSSC}, 2013, pp. 1318--1323.

\bibitem{Yin2014a}
H.~Yin, D.~Gesbert, and L.~Cottatellucci,
\newblock ``A coordinated approach to channel estimation in large-scale
  multiple-antenna systems,''
\newblock {\em {IEEE} J. Sel. Topics Signal Process.}, vol. 8, no. 5, pp.
  942--953, 2014.

\bibitem{Bjornson2015b}
E.~Bj{\"{o}}rnson, M.~Matthaiou, and M.~Debbah,
\newblock ``Massive {MIMO} with non-ideal arbitrary arrays: Hardware scaling
  laws and circuit-aware design,''
\newblock {\em {IEEE} Trans. Wireless Commun.},
\newblock To appear.

\bibitem{Bjornson2014a}
E.~Bj{\"{o}}rnson, J.~Hoydis, M.~Kountouris, and M.~Debbah,
\newblock ``Massive {MIMO} systems with non-ideal hardware: Energy efficiency,
  estimation, and capacity limits,''
\newblock {\em {IEEE} Trans. Inf. Theory}, vol. 60, no. 11, pp. 7112--7139,
  2014.

\bibitem{Pitarokoilis2015a}
A.~Pitarokoilis, S.~K. Mohammed, and E.~G. Larsson,
\newblock ``Uplink performance of time-reversal {MRC} in massive {MIMO} systems
  subject to phase noise,''
\newblock {\em {IEEE} Trans. Wireless Commun.}, vol. 14, no. 2, pp. 711--723,
  2015.

\bibitem{Pitarokoilis2015b}
A.~Pitarokoilis, E.~Bj{\"{o}}rnson, and E.~G. Larsson,
\newblock ``Optimal detection in training assisted {SIMO} systems with phase
  noise impairments,''
\newblock in {\em Proc.~IEEE ICC}, 2015.

\bibitem{Khanzadi2015a}
M.~R. Khanzadi, G.~Durisi, and T.~Eriksson,
\newblock ``Capacity of {SIMO} and {MISO} phase-noise channels with
  common/separate oscillators,''
\newblock {\em {IEEE} Trans. Commun.},
\newblock To appear.

\bibitem{Krishnan2015b}
R.~Krishnan et~al.,
\newblock ``Linear massive {MIMO} precoders in the presence of phase noise---a
  large-scale analysis,''
\newblock {\em {IEEE} Trans. Veh. Technol.},
\newblock To appear.

\bibitem{Gustavsson2014a}
U.~Gustavsson et~al.,
\newblock ``On the impact of hardware impairments on massive {MIMO},''
\newblock in {\em Proc.~IEEE GLOBECOM}, 2014.

\bibitem{Petrovic2007a}
D.~Petrovic, W.~Rave, and G.~Fettweis,
\newblock ``Effects of phase noise on {OFDM} systems with and without {PLL}:
  {C}haracterization and compensation,''
\newblock {\em {IEEE} Trans. Commun.}, vol. 55, no. 8, pp. 1607--1616, 2007.

\bibitem{Zhang2012a}
W.~Zhang,
\newblock ``A general framework for transmission with transceiver distortion
  and some applications,''
\newblock {\em {IEEE} Trans. Commun.}, vol. 60, no. 2, pp. 384--399, 2012.

\end{thebibliography}

\end{document}